\documentclass[journal]{IEEEtran}
\usepackage{algorithm}

\usepackage{algpseudocode}
\usepackage{moreverb}
\usepackage{amsmath}
\usepackage{colortbl}
\usepackage{wrapfig}
\usepackage{bm}
\usepackage{tabularx}
\usepackage{cite}
\usepackage{graphicx}
\usepackage{amsfonts}
\usepackage{amsthm}
\usepackage{amssymb}
\usepackage{cases}
\usepackage{bbm}
\usepackage{stfloats}
\usepackage{subcaption}
\usepackage{xcolor}
\usepackage{xpatch}
\usepackage{epstopdf}
\usepackage[justification=centering]{caption}

\usepackage{balance}

\makeatletter
\ExplSyntaxOn

\ExplSyntaxOff

\newcommand\bib@setcolor[1]{%
	\ifcsname bib@colored@#1\endcsname
	\expandafter\color\expandafter{\csname bib@colored@#1\endcsname}
	\else
	\normalcolor
	\fi
}

\xpatchcmd\@bibitem
{\item}
{\bib@setcolor{#1}\item}
{}{\fail}

\xpatchcmd\@lbibitem
{\item}
{\bib@setcolor{#2}\item}
{}{\fail}
\makeatother

\ifCLASSOPTIONcaptionsoff
\usepackage[nomarkers]{endfloat}
\let\MYoriglatexcaption\caption
\renewcommand{\caption}[2][\relax]{\MYoriglatexcaption[#2]{#2}}
\fi
\usepackage{color}

\newtheorem{Definition}{Definition}
\newtheorem{Theorem}{Theorem}
\newtheorem{Lemma}{Lemma}

\allowdisplaybreaks[4]
\usepackage[
colorlinks=false,     
pdfborder={0 0 0},    
hidelinks             
]{hyperref}

\begin{document}

\title{
  \fontsize{17pt}{19pt}\selectfont 
  Channel Knowledge Aware Finite-Blocklength Rate-Splitting Transmission for High-Mobility Autonomous Driving
}

\author{Yingyang~Chen,~\IEEEmembership{Senior Member,~IEEE},~Yi~Wang,~Yijie Mao,~\IEEEmembership{Senior Member,~IEEE},~Feng~Bai,~Xia~Wu,\\
Feiyan Min,~Li~Wang,~\IEEEmembership{Senior Member,~IEEE},~and Gang~Feng,~\IEEEmembership{Fellow,~IEEE}

\thanks{
}

\thanks{Y. Chen, Y. Wang, F. Bai, X. Wu, and F. Min are with the College of Information Science and Technology, Jinan University, Guangzhou 510632, China (e-mail:~chenyy@jnu.edu.cn;~yiwang@stu2023.jnu.edu.cn;~bbf2392@stu2025.jnu.edu.cn; bobj15594@gmail.com; feiyanmin@jnu.edu.cn).}

\thanks{Y.~Mao is with the School of Information Science and Technology, ShanghaiTech University, Shanghai 201210, China (e-mail: maoyj@shanghaitech.edu.cn).}

\thanks{L.~Wang is with the School of Computer Science, Beijing University of Posts
and Telecommunications, Beijing 100876, China (e-mail: liwang@bupt.edu.cn).}
\thanks{G.~Feng is with the National Key Laboratory of Science and Technology on Communications, University of Electronic Science and Technology of China, Chengdu 611731, China (e-mail: fenggang@uestc.edu.cn).}
    
 }

\markboth{}%
{}


\maketitle

\begin{abstract}

To meet the extended ultra-low latency and high reliability (xURLLC) requirements for autonomous driving systems, multiple access schemes must operate reliably in high-mobility and complex propagation environments. Recently, rate-splitting multiple access (RSMA) has emerged as a promising multi-user transmission framework, showing robustness in dynamic situations where imperfect and outdated channel state information (CSI) is prevalent. 
Moreover, the advanced sensing, localization, and on-board computation capabilities of autonomous driving vehicles facilitate the construction of a channel knowledge map (CKM), which is a key enabler for environment-aware communications in future 6G networks. 
In this work, we propose a CKM-aware finite-blocklength (FBL) RSMA for a downlink autonomous driving system. The location-dependent large-scale channel information provided by CKM is exploited to guide the common/private power split and common-rate allocation in RSMA. 
Specifically, we derive a new and tight closed-form bound for the private-stream ergodic rate. Combined with the closed-form expression for the common-stream ergodic rate, an optimization design of splitting ratios is formulated to maximize the min-rate performance among multiple users. Numerical results show that the proposed scheme outperforms the considered space-division multiple access (SDMA) and non-orthogonal multiple access (NOMA) schemes in high-mobility scenarios. Comparisons across different CKMs further highlight the benefit of accurate channel information, demonstrating how environment-aware channel knowledge complements flexible interference management of RSMA in high-mobility short-packet communications.

\end{abstract}
\vspace{-1mm}
\vspace{-2mm}

%

\section{Introduction}
\vspace{-1mm}

 The rapid development of autonomous systems has given rise to a new paradigm of distributed mobile intelligence, encompassing connected vehicles, autonomous vehicles (AVs), unmanned aerial vehicles (UAVs), and robotic platforms \cite{10138317}. These high-mobility multi-agent systems rely on seamless wireless connectivity to support mission-critical applications such as cooperative perception, collision avoidance, and real-time motion control. Such applications demand frequent exchanges of short packets under stringent latency and reliability constraints \cite{10518091}. Accordingly, efficient resource allocation for ultra-reliable low-latency communication (URLLC) in dynamic wireless environments has attracted increasing research attention \cite{Wu2023URLLC}. In these scenarios, traditional throughput-oriented metrics fail to capture the operational requirements of URLLC. This gap has motivated extended URLLC (xURLLC) in sixth-generation (6G) networks, which is a paradigm tailored to the highly dynamic nature of multi-agent coordination and its associated performance requirements, particularly age of information and round-trip latency \cite{10173693}.


Meeting these stringent performance requirements in highly dynamic environments faces multiple challenges. First, short-packet transmissions are inherently constrained by both block length and target block error rate (BLER), resulting in a pronounced rate loss relative to the classical Shannon capacity \cite{10173693,10106132}. Moreover, in such short-packet scenarios, the relative overhead of pilots for channel estimation increases significantly, and high mobility makes it difficult to acquire accurate channel state information at the transmitter (CSIT), thereby further limiting the effective data rate \cite{9170653}. In addition, ensuring continuous connectivity between moving vehicles and base stations (BSs) is non-trivial in complex propagation environments, where blockage, multi-path fading, and multi-user interference can disrupt communication links \cite{telecom6010013}. 

To address the fundamental challenge of multi-user interference in wireless networks, rate-splitting multiple access (RSMA) has been proposed as a flexible physical-layer strategy. Originally introduced from an information-theoretic perspective to improve the efficiency of multi-user communication \cite{485709}, RSMA has recently evolved into a practical framework for modern multi-antenna systems. Its key principle is to split each user’s message into a common part, decoded by multiple users, and a private part intended for a single user. By allowing part of the interference to be decoded while treating the remaining as noise, RSMA provides a unified interference management mechanism that intelligently bridges space division multiple access (SDMA) and non-orthogonal multiple access (NOMA) \cite{Mao2022RSMA_Survey,Mao2018RSMA_SDMA_NOMA,Clerckx2023RSMA_Primer}. As a result, RSMA has been shown to be robust to channel uncertainty, and is expected to maintain satisfactory performance even when accurate instantaneous CSIT is difficult to obtain \cite{Dizdar2021Mobility}. This robustness makes RSMA a promising candidate for high-mobility multi-agent systems.

While RSMA mitigates the adverse effects of imperfect CSIT, its performance still depends on the availability of channel knowledge. Recently, channel knowledge maps (CKMs) have been proposed as a key enabler for environment-aware communications in future 6G networks \cite{10430216}. Specifically, a CKM is a location-tagged repository that stores prior information about the wireless propagation environment, including large-scale channel characteristics such as path loss, shadowing, blockage probabilities, and other statistical channel properties. By exploiting these slowly varying known parameters inferred from the mobile agent, CKMs enable the transmitter to reduce its reliance on frequent instantaneous pilot-based channel estimation. This capability is particularly beneficial for short-packet transmissions in high-mobility multi-vehicle scenarios, where pilot overhead can otherwise dominate the available resources \cite{11153394}. In extreme cases, where real-time channel estimation yields only marginal performance gains, CKMs can even guide beamforming and transmission decisions without dedicated pilot signals, relying primarily on stored prior information \cite{10108969}.

Notably, the inherent robustness of RSMA to CSIT imperfections makes it a natural candidate to exploit CKM-assisted channel knowledge, even when the stored information does not capture perfect CSI. By combining the interference management flexibility of RSMA with the environment-aware prior knowledge provided by CKMs, the system is expected to achieve improved robustness, reduced pilot overhead, and more reliable performance under short-packet, high-mobility conditions. Therefore, this motivates a systematic investigation of CKM-assisted RSMA transmission in this paper, aiming to enhance performance under stringent xURLLC requirements in realistic propagation environments.

\subsection{Related Work}

\subsubsection{Rate-Splitting Multiple Access}

Extensive studies have demonstrated the benefits of RSMA in multi-user wireless systems, including flexible interference management, enhanced spectral efficiency and user fairness, and robustness to imperfect CSIT~\cite{Mao2022RSMA_Survey,Mao2018RSMA_SDMA_NOMA,Clerckx2023RSMA_Primer,Dizdar2021Mobility}. These advantages have motivated studies on resource allocation and beamforming design, including joint rate and power allocation, globally optimal spectrum- and energy-efficient transmission, and practical max-min fair resource allocation~\cite{Yang2021RSMA,Matthiesen2022RSMA,Luo2023RSMA}. More recently, RSMA has also been investigated across diverse scenarios, including non-terrestrial networks (NTNs)~\cite{10396844}, terahertz communications~\cite{11223134}, integrated sensing and communication (ISAC)~\cite{11316272}, and reconfigurable intelligent surfaces (RISs)~\cite{Chen2026WCL_RSMA_RIS}. These studies highlight the applicability of RSMA to a broad range of communication settings.

Motivated by the latency and reliability requirements of URLLC, recent studies have extended RSMA to short-packet communication (SPC) under finite-blocklength (FBL) constraints~\cite{Xu2022RSMA_FBL_ShortPacket,Xu2023RSMA_FBL_MaxMin} and imperfect CSIT~\cite{10478577,11159579,Zheng2026AutonomousDriving_RSMA_AoI}. Specifically,~\cite{10478577} developed a closed-form ergodic analysis of RSMA-enabled SPC systems by characterizing the common- and private-stream signal-to-interference-plus-noise ratios (SINRs). Subsequently,~\cite{11159579} presented a refined ergodic bound and used it to maximize the ergodic sum rate under quality-of-service (QoS) constraints by jointly optimizing the global power coefficient, private power allocation, and common-rate allocation. RSMA has also been investigated for terahertz SPC~\cite{10251998} and cell-free massive multiple-input multiple-output (MIMO) systems supporting URLLC~\cite{10535307}. These works demonstrate the potential of RSMA under short-packet constraints and imperfect CSIT. Most existing studies, however, focus on small-scale fading imperfections while assuming accurate knowledge of large-scale channel parameters.

\subsubsection{Channel Knowledge Map}

CKMs provide location-dependent channel information, including large-scale fading statistics and site-specific propagation characteristics, to support environment-aware communications~\cite{10430216}. They can be constructed using model-based methods~\cite{9034493}, spatial interpolation~\cite{6333748,7817747}, or data-driven machine learning approaches~\cite{9771802,9354041,9523765}. The resulting channel knowledge can assist channel estimation, beam management, and resource allocation while reducing online acquisition overhead~\cite{10430216}.

CKMs have been studied for a range of communication tasks, including training-free beam alignment for millimeter-wave systems using channel path maps and beam index maps~\cite{9473871}, and multi-UAV placement using Kriging-based map construction~\cite{10272353}. CKM-assisted channel estimation has also been investigated for extremely large antenna array (ELAA) systems, where channel angle-distance maps are combined with a small number of pilot signals to improve estimation accuracy and reduce training overhead~\cite{11153394}. Furthermore, CKMs have been applied to high-mobility scenarios for dual-domain tracking and predictive beamforming, where an extended Kalman filter (EKF) exploits CKM priors to assist angle-of-arrival (AoA) tracking and beam transitions~\cite{11162289}.

These studies demonstrate the usefulness of CKMs for channel acquisition and transmission design. Their system-level evaluations, however, have mainly considered single-user or conventional multiple access settings~\cite{11153394,9473871,10272353,11162289}. In high-mobility multi-user systems, large-scale channel information from a CKM does not by itself provide the current small-scale CSIT needed for spatial interference suppression. Residual multi-user interference remains a limitation of conventional SDMA under channel aging~\cite{Dizdar2021Mobility}. RSMA offers an additional degree of freedom for managing this interference through common/private transmission, while CKM information can guide the associated power and common-rate allocation.


\subsection{Motivations and Contributions}

High-mobility autonomous driving systems require xURLLC services, making FBL transmissions prevalent. Existing FBL-RSMA studies have mainly focused on FBL transmission under perfect or imperfect instantaneous CSI and other physical-layer impairments, while uncertainty in large-scale channel knowledge has generally not been explicitly considered \cite{Xu2022RSMA_FBL_ShortPacket,Xu2023RSMA_FBL_MaxMin,10478577,11159579,Zheng2026AutonomousDriving_RSMA_AoI,10251998,10535307}. To the best of our knowledge, the integration of CKM-assisted channel knowledge into an RSMA-based multi-user framework under FBL constraints and outdated CSIT has not yet been systematically studied.
In such a framework, errors in large-scale channel information affect both private-stream rate predictions and the common-stream bottleneck rate, potentially leading to mismatched power splitting and common-rate allocation. This motivates an investigation of how CKM accuracy affects user fairness in short-packet vehicular communications under channel aging. The main contributions of this paper are summarized as follows

\begin{itemize}
    \item[$\bullet$] 

    We propose a CKM-assisted FBL RSMA framework in which location-dependent large-scale channel information guides the common/private power split and common-rate allocation. The framework considers outdated small-scale CSIT and evaluates the effect of inaccurate CKM information on resource allocation.
    
    \item[$\bullet$] 
    We derive a closed-form lower bound on the ergodic achievable rate of RSMA, explicitly accounting for blocklength, BLER, and power allocation across common and private streams. In particular, the private stream bound is tightened by refining the inequality-based approximation using a closed-form expression for the Shannon-rate term.

    \item[$\bullet$] We formulate a problem that maximizes the minimum ergodic user rate by jointly optimizing the global power coefficient and common-rate allocation. A tractable formulation is obtained using CKM-inferred channel gains and the closed-form rate expressions. For each candidate power coefficient, the common-rate allocation is optimized by water-filling, followed by a one-dimensional grid search over the power coefficient.

    \item[$\bullet$] We evaluate the accuracy of the rate approximation and compare RSMA with SDMA and NOMA under different CKMs. The results characterize how CKM inaccuracies affect the minimum user rate and demonstrate the benefit of accurate large-scale channel information for RSMA in the considered high-mobility scenarios.

\end{itemize}

The rest of the paper is organized as follows. Section~\ref{sec:system} introduces the system model, and Section~\ref{sec:CLOSED-FORM} presents the closed-form ergodic rate analysis. In Section~\ref{problem}, we formulate the problem of maximizing the ergodic min-rate and leverage the closed-form expressions to optimize the allocation of global power and common rate splitting. Section~\ref{result} provides numerical results, and Section~\ref{conclusion} concludes the paper.

\textit{Notation:} Bold lowercase letters denote vectors. The operators $|\cdot|$ and $\|\cdot\|$ indicate set cardinality (or absolute value for scalars) and $l_2$-norm, respectively. $\mathbf{a}^T$ and $\mathbf{a}^H$ are the transpose and Hermitian transpose of $\mathbf{a}$. $\mathcal{CN}(0, \sigma^2)$ represents a circularly symmetric complex Gaussian variable with zero mean and variance $\sigma^2$. ${\mathbf{I}}_N$ is the $N \times N$ identity matrix. $\lfloor \cdot \rfloor$, $\lceil \cdot \rceil$, and $\lfloor \cdot \rceil$ denote floor, ceiling, and rounding to the nearest integer, respectively. $\ln(\cdot)$ is the natural logarithm. $X \sim \mathrm{Gamma}(D, \theta)$ has probability density function (PDF), $f_X(x) = \frac{x^{D-1} e^{-x/\theta}}{\Gamma(D)\theta^D}$ and cumulative distribution function (CDF), $F_X(x) = 1 - \frac{1}{\Gamma(D)} \Gamma(D, x/\theta)$, where the complete Gamma function is $\Gamma(D) = \int_0^\infty t^{D-1} e^{-t} {\rm{d}}t$, its derivative is $\Gamma'(D)$, and the upper incomplete Gamma function is $\Gamma(s,x) = \int_x^\infty t^{s-1} e^{-t} {\rm{d}}t$. The generalized exponential integral is ${E_v}(x) = \int_1^\infty t^{-v} e^{-tx} {\rm{d}}t$ for $x>0$, $v \in \mathbb{R}$. Finally, $\mathbb{E}\{\cdot\}$ denotes expectation.

%
\section{System Model}
\label{sec:system}
\label{Sec_SysMod}

\begin{figure}[t]
       \centering
       \includegraphics[width=1\linewidth]{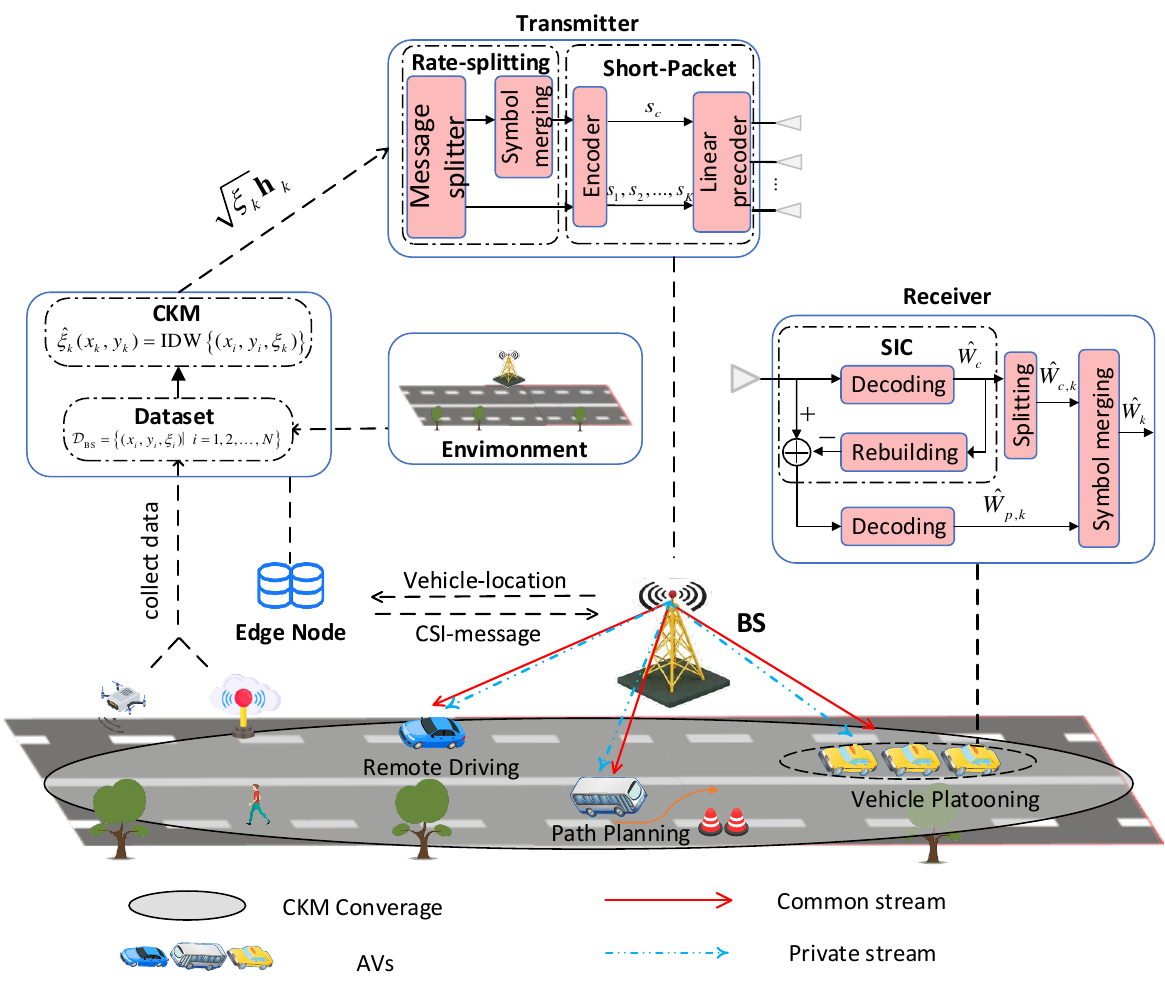}
       \caption{CKM-assisted RSMA short-packet transmission scheme for high-mobility autonomous driving.}
       \label{fig:system_model}
\end{figure}

Fig.~\ref{fig:system_model} illustrates a CKM-assisted high-mobility autonomous driving communication system. The system consists of roadside sensing units (RSUs), edge computing nodes, and connected autonomous vehicles. By collecting channel observations and location information from vehicles through RSUs, together with supplementary measurements acquired by aerial sensing agents in sparsely sampled regions, the edge node continuously constructs and updates a channel knowledge map. We consider a downlink multiple-input single-output (MISO) broadcast channel where a BS equipped with $N_t$ antennas serves $K$ single-antenna vehicles, indexed by the set $\mathcal{K}=\{1,2,\ldots,K\}$, with $N_t \ge K$. Based on the CKM-informed channel characteristics and service requirements, the BS performs message-splitting and power allocation under a single-layer RSMA framework. Specifically, the CKM predicts the rate constraints associated with common and private transmissions. Based on these constraints and service requirements, the BS determines an appropriate rate-splitting strategy and partitions each vehicle's message $W_k$ into a common component $W_{c,k}$ and a private component $W_{p,k}$. The common components are jointly encoded into a common stream $s_c$, while the private components are independently encoded into user-specific private streams ${s_k}$ (${k\in\mathcal{K}}$) for unicast transmission. 
The transmitted signal is given by

\begin{equation}
\label{eq:transmitted_signal}
\mathbf{x}=\sqrt{P(1-t)}\mathbf{p}_{c}s_{c}+\sqrt{Pt}\sum_{k\in\mathcal{K}}\sqrt{\mu_{k}}\mathbf{p}_{k}s_{k},
\end{equation}
where $\mathbf{p}_{c}\in\mathbb{C}^{N_t\times 1}$ and $\mathbf{p}_{k}\in\mathbb{C}^{N_t\times 1}$ are the precoders for the common stream and the $k$-th private stream, respectively, satisfying $\|\mathbf{p}_{c}\|^{2}=\|\mathbf{p}_{k}\|^{2}=1,\forall k\in\mathcal{K}$. Here $P$ represents the total transmit power, and $\mathbb{E}\left\{\mathbf{s}\mathbf{s}^{H}\right\}=\mathbf{I}$. The global power coefficient $0\leq t\leq 1$ governs the power division between common and private streams, while the private power distribution $0\leq\mu_{k}\leq 1$ with $\sum_{k=1}^{K}\mu_{k}=1$ specifies power allocation across private streams.

The received signal at vehicle-$k$ is
\begin{equation}
\label{eq:received_signal}
y_{k}=\sqrt{\xi_{k}}\mathbf{h}_{k}^{H}[m]\mathbf{x}+z_{k}, \quad k\in\mathcal{K},
\end{equation}
where $\mathbf{h}_{k}[m]\in\mathbb{C}^{N_t\times 1}$ denotes the small-scale fading, $z_{k}\sim\mathcal{CN}(0,\sigma^{2})$ is additive white Gaussian noise (AWGN), and $\xi_{k}$ represents large-scale fading. 
For analytical convenience, we normalize the noise by scaling the received signal:
\begin{equation}
y_k^{\text{scaled}} = \sqrt{\zeta_k}\mathbf{h}_k^H[m]\mathbf{x} + n_k, \quad k \in \mathcal{K},
\end{equation}
where $\zeta_k = \xi_k/\sigma^2$ and $n_k \sim \mathcal{CN}(0,1)$.

In the RSMA framework, each vehicle first decodes the common stream by treating all private streams as interference. The signal-to-interference-plus-noise ratio (SINR) for the common stream at vehicle-\( k \) can be expressed as:

\begin{align}
    \label{SINR_c}
    {\Gamma _{c,k}} = \frac{{{P}{(1-t)}{\zeta_{k}}{{\left| {{\mathbf{h}}_k^H[m]{{\mathbf{p}}_c}} \right|}^2}}}{{{P}{t}{\zeta_{k}}\sum_{j \in {\mathcal K}} {{\mu_j}{{\left| {{\mathbf{h}}_k^H[m]{{\mathbf{p}}_j}} \right|}^2} + 1} }}. 
\end{align}
Each vehicle implements SIC to decode the common stream $s_c$ and subtract it from the received signal $y_k$. 
To guarantee successful decoding of $s_c$ by all vehicles, the design focuses on the SINR of the user with the minimum SINR, i.e., $\min_{k\in\mathcal{K}} \Gamma_{c,k}$. For analytical tractability, perfect SIC is assumed. After successfully decoding and removing the common stream via SIC, each vehicle decodes its intended private stream. The SINR for the private stream at vehicle-\( k \) is given by:

\begin{align}
    \label{SINR_p}
    {\Gamma _{p,k}} = \frac{{{P}{t}{\zeta_{k}}{\mu_k}{{\left| {{\mathbf{h}}_k^H[m]{{\mathbf{p}}_k}} \right|}^2}}}{{{P}{t}{\zeta_{k}}\sum\nolimits_{j \in {\mathcal K},j \ne k} {{\mu_j}{{\left| {{\mathbf{h}}_k^H[m]{{\mathbf{p}}_j}} \right|}^2} + 1} }}. 
\end{align}
Following the decoding of both the common and private streams, vehicle $k$ reconstructs the original message by combining the decoded common component $\hat{W}_{c,k}$ with the decoded private component $\hat{W}_{p,k}$.
In the FBL regime, the Shannon capacity no longer accurately characterizes achievable rates. We therefore adopt the normal approximation \cite{5452208}, which yields the instantaneous rates of the common and private streams as
\begin{subequations}\label{eq:5}
\begin{align}
R_{c}^{(m)} &\approx 
C\!\left(\min_{k\in\mathcal{K}}\Gamma_{c,k}\right)
-\sqrt{\frac{V(\min_{k\in\mathcal{K}}\Gamma_{c,k})}{l_{c}}}\,
Q^{-1}(\beta_{c,k}), \label{eq:5a}\\
R_{p,k}^{(m)} &\approx 
C(\Gamma_{p,k})
-\sqrt{\frac{V(\Gamma_{p,k})}{l_{k}}}\,
Q^{-1}(\beta_{p,k}), \label{eq:5b}
\end{align}
\end{subequations}where $l_{c}$ and $l_{k}$ denote the blocklengths, $\beta_{c,k}$ and $\beta_{p,k}$ are the block error rates (BLERs), $C(x)=\log_{2}(1+x)$ is the Shannon capacity function, and $V(x) = (1 - (1 + x)^{-2})(\log_2e)^2$ is the channel dispersion. Therefore, the corresponding ergodic rates for the common and private streams are given by
$R_{c}=\mathbb{E}\!\left\{R_{c}^{(m)}\right\}$ and 
$R_{k}=\mathbb{E}\!\left\{R_{p,k}^{(m)}\right\}$, respectively.

\section{Channel Knowledge Construction}

This section introduces the channel modeling framework for CKM-assisted vehicular communications, capturing both large-scale and small-scale fading with practical imperfections. The transmitter relies on two sources of prior channel knowledge. Firstly, location-dependent large-scale fading information is obtained from a pre-constructed CKM based on estimated vehicle positions. These position estimates are assumed to be available at the BS via existing localization mechanisms, such as global navigation satellite system (GNSS) reports and network-based positioning \cite{electronics11193247}.
Secondly, delayed small-scale fading information from the previous transmission interval is assumed to be available at the transmitter. This prior knowledge can be obtained through low-overhead techniques, such as channel prediction or by exploiting side information, which require minimal or no pilot assistance.
By considering both sources of prior knowledge, the system evaluation captures the effects of errors in both large-scale and small-scale channel information on transmission performance. This information can be used to guide power allocation and message-splitting decisions.

\begin{figure}[t]\centerline{\includegraphics[width=0.40\textwidth,keepaspectratio]{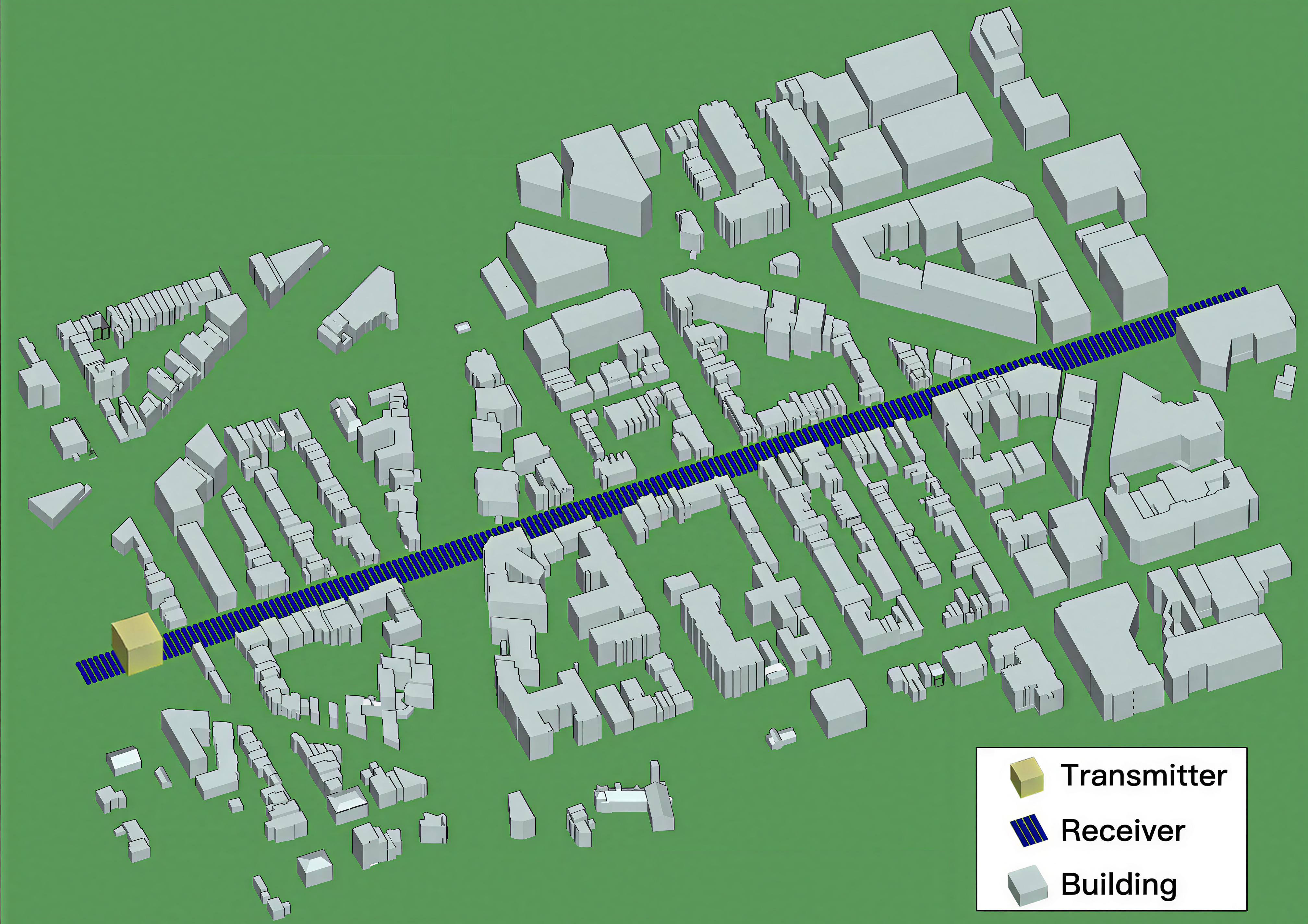}}
	\vspace{2mm}
	\caption {Urban road segment layout with sampled transmitter and receiver locations.}
	\label{road}
\end{figure}

\begin{figure}[t]
	\centering
	\includegraphics[width=0.45\textwidth]{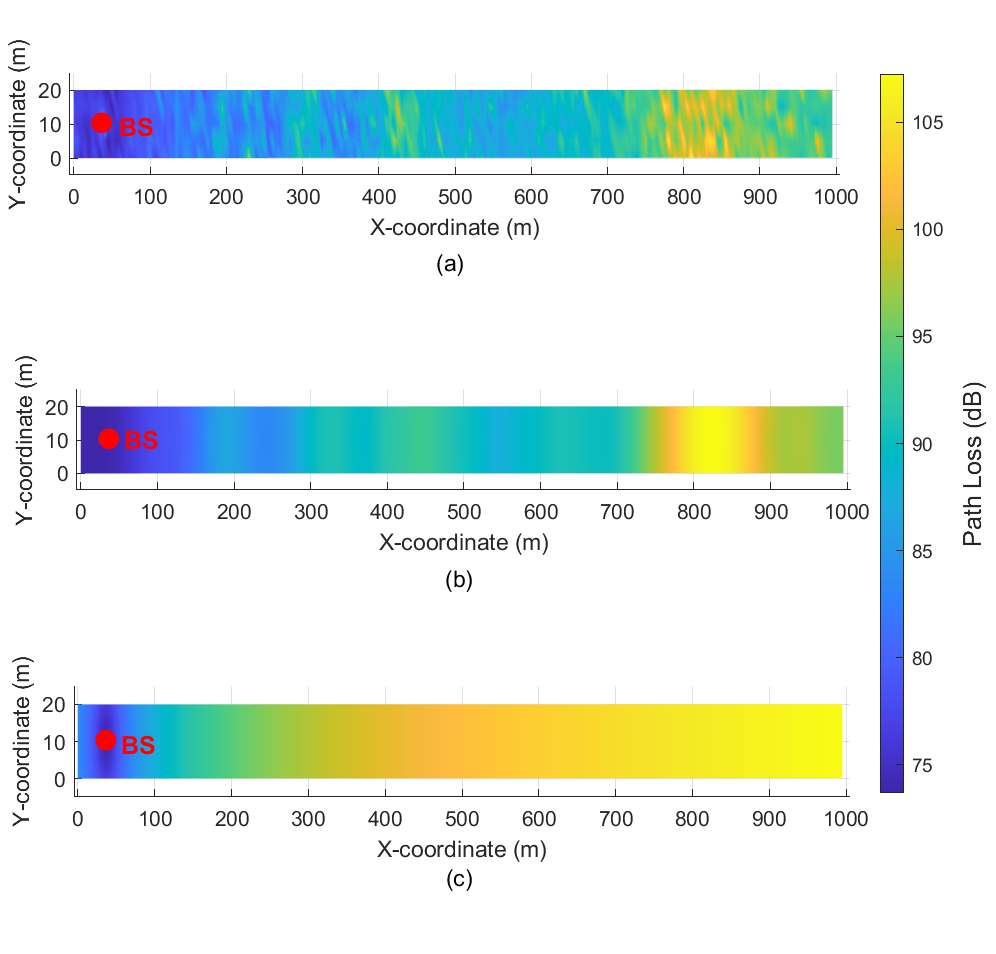} 
	\caption{Path-loss maps considered: (a) ground-truth CKM from dense 5~m sampling; (b) KNN-interpolated CKM from sparse 25~m measurements; (c) path-loss CKM.}
	\label{CKM_total}
	\vspace{3mm}
\end{figure}

\subsection{Large-Scale Channel Knowledge Construction}

To evaluate ergodic achievable rates in an urban vehicular environment and compare different CKM construction methods, a site-specific propagation scenario is modeled in Remcom Wireless InSite~\cite{remcom_wireless_2021}, as illustrated in Fig.~\ref{road}. The ray-tracing simulator generates environment-aware channel realizations that serve as a quasi-ground-truth reference for CKM evaluation. In practical deployments, such channel measurements are acquired from uplink pilot estimates collected
by autonomous veihcile in conjunction with positioning information provided by roadside sensing units. In this scenario, the yellow module represents the transmitter, the blue modules denote densely deployed sampled receivers along a straight $1000~\mathrm{m} \times 25~\mathrm{m}$ road segment, and the gray structures correspond to surrounding buildings. The presence of the line-of-sight (LoS) path ensures that conventional distance-based path-loss models, while approximate, remain reasonably representative of the overall large-scale attenuation. 

Based on the ray-tracing results, a high-resolution \emph{path-loss map} is generated by sampling the received power every $\Delta_1 = 5~\mathrm{m}$ along the road, corresponding to the dense receiver positions in Fig.~\ref{road}. This densely sampled map, shown in Fig.~\ref{CKM_total}(a), serves as the \emph{ground-truth CKM}, which can be represented as a function $L(x,y)$ mapping a spatial location $(x,y)$ to the corresponding large-scale path loss in dB. The path-loss values exhibit pronounced location-dependent variations, reflecting shadowing and site-specific propagation effects beyond a simple distance-dependent trend.

In practice, dense channel measurements are rarely available due to limited measurement density and localization errors. To model realistic CKM acquisition, we assume that the transmitter only has access to sparse path-loss measurements at locations spaced by $\Delta_2 = 25~\mathrm{m}$. CKMs can, in general, be constructed from spatially sampled channel measurements or sensing data using model-based approaches~\cite{9034493}, interpolation techniques~\cite{6333748,7817747}, or data-driven methods based on machine learning~\cite{9771802,9354041,9523765}. In this work, a continuous CKM over the entire road segment is reconstructed using a K-nearest neighbor (KNN) interpolation method, selected for its low computational complexity and analytical transparency. For each spatial location without direct measurement, the Euclidean distances $d_i$ to its five nearest sampled neighbors are computed, and the interpolation weights are defined as
\begin{equation}
w_i = \frac{\exp(-d_i^2 / \sigma)}{\sum_{j=1}^{5} \exp(-d_j^2 / \sigma)}, 
\quad 
\sigma = \frac{1}{5}\sum_{i=1}^{5} d_i^2.
\end{equation}
The resulting KNN-interpolated path-loss map, shown in Fig.~\ref{CKM_total}(b), captures the overall large-scale trend, including periodic variations and a gradual decay with distance.

As a representative model-based CKM, a conventional analytical map is constructed using a distance-dependent path-loss model, where the path loss at location $k$ is expressed as
\begin{equation}
\label{pathloss_model}
l(d_k) = p_0 + 10\alpha \lg(d_k),
\end{equation}
with $d_k$ denoting the BS--receiver distance. In our work, the parameters $p_0$ and $\alpha$ are calibrated using measured path losses at reference distances of $100~\mathrm{m}$ and $900~\mathrm{m}$ to align the analytical model with the simulated propagation scenario. The resulting path-loss CKM, shown in Fig.~\ref{CKM_total}(c), exhibits a smooth and monotonic decay with distance.

\subsection{Small-Scale Channel Knowledge and Imperfect CSIT Model}

To model small-scale fading and imperfect CSIT, we adopt the autoregressive fading model in~\cite{5723047}. The channel of vehicle~$k$ at time~$m$ is expressed as
\begin{equation}
\label{eq:channel_model}
\mathbf{h}_k[m] = \varepsilon \, \mathbf{h}_k[m-1] + \sqrt{1-\varepsilon^2} \, \mathbf{e}_k[m],
\end{equation}
where $\mathbf{e}_k[m]$ is a random component following spatially uncorrelated Rayleigh fading with i.i.d.\ $\mathcal{CN}(0,1)$ entries. The temporal correlation coefficient
\(
\varepsilon = J_0(2\pi f_{D,k}T)
\)
follows the Jakes model, with $T$ denoting the channel coherence interval and $f_{D,k} = v_k f_c / c$ the maximum Doppler frequency determined by vehicle speed $v_k$, carrier frequency $f_c$, and the speed of light $c$.

The previous channel realization $\mathbf{h}_k[m-1]$ is assumed to be available at the transmitter. In practice, it can be obtained with limited or no pilot overhead using CKMs \cite{10430216}, integrated sensing and communication \cite{10536135}, or autoregressive prediction based on historical measurements \cite{9210016}. In Monte Carlo simulations, $\mathbf{h}_k[m-1]$ is directly generated as an independent complex Gaussian random vector with identical second-order statistics, which provides a reasonable approximation for high-mobility urban vehicular environments~\cite{5723047}. In this model, the error from acquiring small-scale fading via CKMs or other methods is considered negligible relative to Doppler-induced channel variations, which is a reasonable approximation in high- or ultra-high-mobility scenarios.

\section{Closed-Form Ergodic Rate Analysis}
\label{sec:CLOSED-FORM}
This section presents the derivation of a closed-form lower bound on the ergodic rate for transmission streams under general antenna configurations, BLER constraints, and finite blocklength. To facilitate the subsequent analysis, the distributions and approximations of several frequently used random variables are first summarized. Based on these results, we first present the closed-form expression for the common-stream rate using established formulations. We then derive a tighter closed-form lower bound on the private-stream rate, providing an improved characterization at finite blocklength.

\subsection{Gamma Distribution for Precoding}

For the private streams, ZF precoding is employed. Since the ZF beamformer $\mathbf{p}_k$ is isotropically distributed and independent of the Gaussian error vector $\mathbf{e}_j^H[m]$ for all $j \in \mathcal{K}\setminus k$, the term $|\mathbf{e}_j^H[m]\mathbf{p}_k|^2$ follows an exponential distribution with unit mean, i.e., $|\mathbf{e}_j^H[m]\mathbf{p}_k|^2 \sim \mathrm{Gamma}(1,1)$.  
Because the transmitter only has access to the previous-slot CSI $\mathbf{h}_k[m-1]$, it follows that $\mathbf{h}_j^H[m-1]\mathbf{p}_k = 0$ for all $j \in \mathcal{K}\setminus k$, and the desired-signal component $|\mathbf{h}_k^H[m-1]\mathbf{p}_k|^2$ follows $\mathrm{Gamma}(N_t - K + 1, 1)$~\cite{6967804}. The common-stream precoder $\mathbf{p}_c$ is modeled as a stochastic beamformer uncorrelated with $\mathbf{h}_k[m-1]$ and $\mathbf{e}_k[m]$, resulting in $|\mathbf{h}_k^H[m]\mathbf{p}_c|^2 \sim \mathrm{Gamma}(1,1)$~\cite{7152864}. For simplification, each term $\left|\mathbf{h}_k^H[m]\mathbf{p}_j\right|^2$ with $j,k \in \mathcal{K}$ is approximated as
\begin{align}
\label{Approximation1}
\left|\mathbf{h}_k^H[m]\mathbf{p}_k\right|^2 &\approx \varepsilon^2 \left|\mathbf{h}_k^H[m-1]\mathbf{p}_k\right|^2 + (1-\varepsilon^2) \left|\mathbf{e}_k^H[m]\mathbf{p}_k\right|^2, \\
\label{Approximation2}
\left|\mathbf{h}_k^H[m]\mathbf{p}_j\right|^2 &\approx (1-\varepsilon^2) \left|\mathbf{e}_k^H[m]\mathbf{p}_j\right|^2, \quad j \neq k.
\end{align}
These approximations ignore the cross terms involving both $\mathbf{h}_k[m-1]$ and $\mathbf{e}_k[m]$ \cite{6967804}.

Since many involved random variables are Gamma distributed, we apply the approximation method in~\cite{6967804} for the sum \( Z=\sum_i A_i \) with \( A_i \sim \mathrm{Gamma}(D_{A_i},\theta_{A_i}) \).  
Under this method, \( Z \) is represented by an equivalent Gamma r.v. \( \widetilde{Z} \sim \mathrm{Gamma}(\widetilde D_Z,\widetilde\theta_Z) \), whose parameters are given by
\begin{align}
\label{gammajinsi}
\widetilde D_Z = 
\frac{\left(\sum_i D_{A_i}\theta_{A_i}\right)^2}
     {\sum_i D_{A_i}\theta_{A_i}^2},
\qquad
\widetilde\theta_Z =
\frac{\sum_i D_{A_i}\theta_{A_i}^2}
     {\sum_i D_{A_i}\theta_{A_i}}.
\end{align}

\subsection{Lower Bound for $R_c$}

The closed-form lower bound for the ergodic rate of the common stream $R_c$ is
\begin{equation}
\label{Rc_lower_bound_final}
\widehat{R}_c =
\frac{e^{\frac{C_1}{C_2}}}{\ln 2}\,\Psi(D)
-
\sqrt{
\frac{
V\!\Big\{\frac{e^{\frac{C_1}{C_2}}}{C_2}\,E_D\Big(\frac{C_1}{C_2}\Big) \Big\}
}{l_c}
} \, Q^{-1}(\beta_{c,k}),
\end{equation}
where the parameters are $C_1 = \sum_{k=1}^K \frac{1}{\zeta_k P(1-t)}$, $C_2 = \frac{\widetilde{\theta}\, t}{K(1-t)}$, and $D = \lfloor \widetilde{D}\, K \rceil$, with $\widetilde D = \frac{[\varepsilon^2(N_t+1) + (1-2\varepsilon^2) K]^2}{\varepsilon^4(N_t+1) + (1-2\varepsilon^2) K}$ and $\widetilde \theta = \frac{\varepsilon^4(N_t+1) + (1-2\varepsilon^2) K}{\varepsilon^2(N_t+1) + (1-2\varepsilon^2) K}$. The composite exponential integral $
\Psi(D)=
\frac{e^{\frac{C_1}{C_2}(C_2-1)}}{(1-C_2)^D}E_1(C_1)
-\sum_{i=1}^{D}
\frac{E_i\!\left(\frac{C_1}{C_2}\right)}
{(1-C_2)^{D+1-i}}$, 
with the generalized exponential integral $E_v(x) = \int_1^\infty e^{-t x} t^{-v}{\rm{d}}t$.


\begin{proof}
The derivation follows a similar procedure to that in Section III-B of~\cite{11159579}. To make this paper self-contained, a sketch of the proof is provided in Appendix A, while the complete algebraic derivation can be found in~\cite{11159579}.
\end{proof}

\subsection{Lower Bound for \( R_k \)}
Then we derive the closed-form lower bound for the ergodic rate of the $k$-th private stream, $R_k$. Unlike \cite{Dizdar2021Mobility} and \cite{11159579}, which lower-bound the Shannon-rate terms using Jensen-type inequalities, we derive a tighter closed-form approximation by directly evaluating their expectations based on Gamma-distributed approximations. Specifically, we first rewrite 
\begin{align}
\mathbb{E}[C(\Gamma_{p,k})] 
&= \mathbb{E}\Big[\log_2\Big(1 + Pt \zeta_k \sum_{j\in\mathcal{K}} \mu_j |\mathbf{h}_k^H[m]\mathbf{p}_j|^2\Big)\Big] \nonumber\\
&\quad - \mathbb{E}\Big[\log_2\Big(1 + Pt \zeta_k \sum_{j\in\mathcal{K}, j\neq k} \mu_j |\mathbf{h}_k^H[m]\mathbf{p}_j|^2\Big)\Big]
\label{eq_ergodic_rate_decomposition}.
\end{align}
To obtain a concise closed-form approximation, we introduce the following lemmas.
\begin{Lemma}
	\label{Xkdejinsi}
	The r.v. $X_{k} = \sum\nolimits_{j \in {\mathcal K}} {{\mu_j}{{\left| {{\mathbf{h}}_k^H[m]{{\mathbf{p}}_j}} \right|}^2}}$ is approximated by the r.v. $\widetilde X_k$ with the distribution $Gamma({\widetilde D_{X_k}},{\widetilde \theta _{X_k}})$, where
	\begin{align}
		\label{dxkthxk}
		{\widetilde D_{X_k}} = \frac{(({N_t-K+1}){\varepsilon^2}{\mu_{k}}+1-{\varepsilon^2})^2}{(N_t-K+1)\varepsilon^4\mu_{k}^2+{(1-\varepsilon^2)^2}{\sum\nolimits_{j \in {\mathcal K}} {{\mu_j}^2} }},\hfill\\{\widetilde\theta_{X_k}}=\frac{{(N_t-K+1)\varepsilon^4\mu_{k}^2+{(1-\varepsilon^2)^2}{\sum\nolimits_{j \in {\mathcal K}} {{\mu_j}^2} }}}{({N_t-K+1}){\varepsilon^2}{\mu_{k}}+1-{\varepsilon^2}},
	\end{align}
\end{Lemma}
\begin{Lemma}
	\label{Ykdejinsi}
	The r.v. $ Y_k= \sum\nolimits_{j \in {\mathcal K}\backslash k} {{\mu_j}{{\left| {{\mathbf{h}}_k^H[m]{{\mathbf{p}}_j}} \right|}^2}} $ is approximated by the r.v. $\widetilde Y_k$ with the distribution $Gamma({\widetilde D_{Y_k}},{\widetilde \theta _{Y_k}})$, where
\begin{align}
\label{d4o4}
{\widetilde D_{Y_k}} = \frac{{{{(1 - {\mu_k})}^2}}}{{\sum\nolimits_{j \in {\mathcal K}\backslash k} {{\mu_j}^2} }},\ {\widetilde \theta _{Y_k}} = \frac{{(1-\varepsilon^2)}{\sum\nolimits_{j \in {\mathcal K}\backslash k} {{\mu_j}^2} }}{{1 - {\mu_k}}},
\end{align}
\begin{proof}
The proofs of Lemmas~\ref{Xkdejinsi} and~\ref{Ykdejinsi} are similar to that of~\cite[Lemma~5]{11159579} and are omitted here for brevity.
\end{proof}

\end{Lemma}

\begin{Lemma}
	\label{log1plusZ}
	Let $\widetilde Z$ follow $\mathrm{Gamma}(\widetilde D_Z, \widetilde \theta_Z)$ with integer $\widetilde D_Z \ge 1$ and $\widetilde \theta_Z > 0$. Then
	\[
	\mathbb{E}[\log_2(1 + \widetilde Z)] = \sum_{i=0}^{\widetilde D_Z-1} \frac{e^{1/\widetilde \theta_Z} E_{i+1}\left( \frac{1}{\widetilde \theta_Z} \right)}{\ln 2},
	\]
	where $E_m(x) = \int_1^\infty t^{-m} e^{-xt} {\rm{d}}t$ is the exponential integral function.
	\begin{proof}
    The detailed proof is deferred to Appendix~B for brevity, as this lemma is a fundamental step in the derivations.
	\end{proof}
\end{Lemma}

\begin{Definition}
\label{E_log2_gamma_approximation}
For a Gamma-distributed random variable 
$\widetilde Z\sim\mathrm{Gamma}(\widetilde D_Z,\widetilde\theta_Z)$
with a possibly non-integer shape parameter $\widetilde D_Z>0$, 
we approximate 
\(
\mathbb{E}[\log_2(1+\widetilde Z)] \approx \Phi(\widetilde D_Z,\widetilde\theta_Z),
\)
where
\begin{align}
\Phi(\widetilde D_Z,\widetilde\theta_Z) \triangleq\;&
w_Z \sum_{i=0}^{\lfloor \widetilde D_Z \rfloor - 1} 
    \frac{e^{1/\widetilde\theta_Z} E_{i+1}(1/\widetilde\theta_Z)}{\ln 2} 
 \nonumber\\&
+ (1-w_Z)\sum_{i=0}^{\lceil \widetilde D_Z \rceil - 1} 
    \frac{e^{1/\widetilde\theta_Z} E_{i+1}(1/\widetilde\theta_Z)}{\ln 2},
\end{align}
and the weight is given by \(
w_Z = \lceil \widetilde D_Z \rceil - \widetilde D_Z.
\)

This approximation uses the exact integer-case expression in 
Lemma~\ref{log1plusZ} at the lower and upper integer bounds 
$\lfloor \widetilde D_Z \rfloor$ and $\lceil \widetilde D_Z \rceil$, 
and interpolates between them by distance-weighted averaging.
\end{Definition}

\begin{Lemma}
\label{lem_Epk_approx}
The expectation of $\Gamma_{p,k}$ is approximated by
\[
\mathbb{E}[\Gamma_{p,k}] \approx \mathbb{E}[\widetilde{\Gamma}_{p,k}] = \frac{e^{C_{k_1}/C_{k_2}}}{C_{k_2}} E_{[K-1]}\Big(\frac{C_{k_1}}{C_{k_2}}\Big),
\]
where
$C_{k_1} = \frac{1}{P t \zeta_k \mu_k \widetilde{D}_{M_k} \widetilde{\theta}_{M_k}}$, $
C_{k_2} = \frac{(1-\mu_k)(1-\varepsilon^2)}{(K-1) \mu_k \widetilde{D}_{M_k} \widetilde{\theta}_{M_k}},$
$\widetilde{D}_{M_k} = \frac{[(N_t - K + 1)\varepsilon^2 + (1-\varepsilon^2)]^2}{(N_t - K + 1)\varepsilon^4 + (1-\varepsilon^2)^2}$, $
\widetilde{\theta}_{M_k} = \frac{(N_t - K + 1)\varepsilon^4 + (1-\varepsilon^2)^2}{(N_t - K + 1)\varepsilon^2 + (1-\varepsilon^2)}.$

\end{Lemma}
\begin{proof}
The proof of Lemmas~\ref{lem_Epk_approx} is similar to Lemma~6 in~\cite{11159579}, and is omitted here due to page limitations.
\end{proof}

\begin{Theorem}
The lower bound \( \widehat{R}_k \) for the ergodic rate of the $k$-th private stream \( R_k \) is given by
\begin{align}
\label{R_k_bound_Phi}
\nonumber 
\widehat{R}_k(t,\boldsymbol{\mu})
&= 
\Phi(\widetilde D_{X_k},\alpha_k)
- 
\Phi(\widetilde D_{Y_k},\eta_k)
\\[1mm]
&\quad
{-} \sqrt{
	\frac{
		V\!\left(
		\frac{e^{C_{k_1}/C_{k_2}}}{C_{k_2}}
		E_{[K-1]}\!\left(
		\frac{C_{k_1}}{C_{k_2}}
		\right)
		\right)
	}{l_k}
}\,
Q^{-1}(\beta_{p,k}),
\end{align}
where 
\(\alpha_k = Pt \, \zeta_k \, \widetilde\theta_{X_k}\) and 
\(\eta_k = Pt \, \zeta_k \, \widetilde\theta_{Y_k}\), 
with \(\widetilde D_{X_k}\) and \(\widetilde \theta_{X_k}\) given in Lemma~\ref{Xkdejinsi}, \(\widetilde D_{Y_k}\) and \(\widetilde \theta_{Y_k}\) given in Lemma~\ref{Ykdejinsi}; 
\(C_{k_1}\) and \(C_{k_2}\) are defined in Lemma~\ref{lem_Epk_approx}; 
and \(\Phi(\cdot,\cdot)\) is defined in Definition~\ref{E_log2_gamma_approximation}.

\end{Theorem}

\begin{proof}
We first apply the linearity of expectation to rewrite \( R_k \) as
\begin{align}
\nonumber R_k &= \mathbb{E}\left[ \log_2\left( 1+ Pt\, \zeta_{k} X_k \right) \right] 
- \mathbb{E}\left[ \log_2\left( 1+ Pt\, \zeta_{k} Y_k \right) \right] \\
& \quad - \mathbb{E}\left[ \sqrt{\frac{V(\Gamma_{p,k})}{l_k}} Q^{-1}(\beta_{p,k}) \right].
\end{align}

By approximating 
\(
X_k \approx \widetilde X_k \sim \mathrm{Gamma}(\widetilde D_{X_k}, \widetilde \theta_{X_k})
\) and 
\(
Y_k \approx \widetilde Y_k \sim \mathrm{Gamma}(\widetilde D_{Y_k}, \widetilde \theta_{Y_k})
\),
we define the scaled variables
\[
\alpha_k \triangleq Pt \zeta_k \, \widetilde \theta_{X_k}, \quad 
\eta_k \triangleq Pt \zeta_k \, \widetilde \theta_{Y_k},
\]
which corresponds to multiplying the original Gamma variables by the constant \(Pt \zeta_k\). By applying the weighted integer-bound approximation in Definition~\ref{E_log2_gamma_approximation}.
Therefore,
\begin{align}
R_k &\approx \Phi(\widetilde D_{X_k},\alpha_k) - \Phi(\widetilde D_{Y_k},\eta_k) 
- \mathbb{E}\left[ \sqrt{\frac{V(\Gamma_{p,k})}{l_k}} Q^{-1}(\beta_{p,k}) \right].
\end{align}

Since \(\sqrt{V(\cdot)}\) is concave, applying Jensen's inequality gives the lower bound
\begin{align}
R_k &\ge \Phi(\widetilde D_{X_k},\alpha_k) - \Phi(\widetilde D_{Y_k},\eta_k)
- \sqrt{\frac{V(\mathbb{E}[\Gamma_{p,k}])}{l_k}}\, Q^{-1}(\beta_{p,k}).
\end{align}

Finally, using Lemma~\ref{lem_Epk_approx} to approximate \(\mathbb{E}[\Gamma_{p,k}] \approx \mathbb{E}[\widetilde{\Gamma}_{p,k}]\), we obtain the lower bound \(\widehat{R}_k\) as stated in \eqref{R_k_bound_Phi}.
\end{proof}

\vspace{-2mm}
\section{Problem Formulation and Solution}
\label{problem}

This section focuses on maximizing the minimum achievable user rate in RSMA-based FBL transmission. To simplify the problem and enable global optimization, the private-stream powers are equally allocated, i.e., $\mu_k = 1/K$. The optimization that jointly considers the global power coefficient \(t\) and the common-rate splitting fractions \(\mathbf c=[c_1,\ldots,c_K]^T\) to maximize the min-rate performance among $K$ users is formulated as
\begin{subequations}\label{eqn:problem_p0}
\begin{align}
\mathcal{P}_1: &\ \max_{t,\, \mathbf{c}} \ \min_{k \in \mathcal{K}} \big(R_c(\zeta_1,\dots,\zeta_K)c_k+ R_k(\zeta_k)\big) \\[1mm]
\text{s.t.}\;& \sum_{k \in \mathcal{K}} c_k = 1, \\
& c_k \ge 0, \quad \forall k \in \mathcal{K}, \\
& 0\le t\le1.
\end{align}
\end{subequations}
Here $\zeta_k$ represents the true large-scale channel power gain of user~$k$, which is obtained from the ground-truth CKM via $\zeta_k = \frac{10^{-L(x_k,y_k)/10}}{\sigma^2}$, where $L(x_k,y_k)$ is the path loss at the user location $(x_k,y_k)$ and $\sigma^2$ is the receiver noise power. In $\mathcal{P}_1$, \(R_c(\zeta_1,\dots,\zeta_K)=\mathbb{E}\{R_c^{(m)}(\zeta_1,\dots,\zeta_K)\}\) and \(R_k(\zeta_k)=\mathbb{E}\{R_{p,k}^{(m)}(\zeta_k)\}\) denote the original ergodic rates defined at the end of Section II, which generally require numerical expectation evaluation.

In the considered model, the BS obtains estimates of the large-scale
channel gains via a CKM constructed from the sensed vehicle
positions. The true gains $\zeta_k$ are not directly available, and the
transmitter instead relies on the CKM-inferred values \(\widehat{\zeta}_k\) for
optimization. Importantly, in RSMA, the CKM-inferred gains affect not only the private-stream rates but also the common-stream rate and its allocation among users. Hence, CKM errors may propagate into both the common/private power split and the common-rate allocation, leading to mismatched resource allocation. Moreover, direct evaluation of the ergodic rates \(R_c\) and \(R_k\) in \(P_1\) requires numerical expectation. To obtain a tractable CKM-assisted formulation, we replace them with the closed-form expressions  $\widehat{R}_c$ in \eqref{Rc_lower_bound_final} and
$\widehat{R}_k$ in \eqref{R_k_bound_Phi}, respectively, evaluated using the CKM-inferred gains. Accordingly, \(\mathcal{P}_2\) is formulated as
\begin{subequations}\label{eqn:problem_p2}
\begin{align}
\mathcal{P}_2: &\ \max_{t,\, \mathbf{c}} \ \min_{k \in \mathcal{K}} 
\big(\widehat{R}_c(\widehat{\zeta}_1, \dots, \widehat{\zeta}_K)c_k + \widehat{R}_k(\widehat{\zeta}_k)\big) \\
\text{s.t.}\;& \sum_{k \in \mathcal{K}} c_k = 1,\\
& c_k \ge 0, \quad \forall k \in \mathcal{K}, \\
& 0\le t\le1.
\end{align}
\end{subequations}

To solve $\mathcal{P}_2$, we search over a finite set of candidate values of the scalar variable~$t$, which guarantees that the global optimum is found. For each candidate value of $t$, the common-rate splitting subproblem is solved optimally via a water-filling procedure. The resulting solution is optimal over the evaluated candidates, while its accuracy relative to the continuous-domain optimum depends on the search-grid resolution. The water-filling procedure is described as follows.
The private-stream rates are first sorted
in ascending order, $\widehat{R}_{(1)} \le \dots \le \widehat{R}_{(K)}$, and for
each $m = 1,2,\dots,K$, the corresponding water level is computed as
\[
s_m = \frac{\widehat{R}_c(t) + \sum_{i=1}^{m} \widehat{R}_{(i)}(t)}{m}.
\]
The valid $m$ is chosen as the first index satisfying
$\widehat{R}_{(m)} \le s_m \le \widehat{R}_{(m+1)}$, ensuring that no user's
allocated rate falls below its private-stream rate. For convenience, let \(C_k\) denote the common rate assigned to user \(k\) for a given \(t\). The water-filling solution yields
$C_k = \max\{0, s_m - \widehat{R}_k(t)\}$, and hence
$c_k = C_k / \widehat{R}_c(t)$, guaranteeing $\sum_{k=1}^K c_k = 1$.

The resulting minimum achievable user rate for this candidate $t$ is
$s^*(t) = \min_k \big(C_k + \widehat{R}_k(t)\big)$. After evaluating all
candidate $t$ values, the optimal power coefficient and corresponding common-rate allocation over the search grid are obtained as $t^* = \arg\max_t s^*(t)$ and $\mathbf{c}^* = \mathbf{c}(t^*)$,
which together provide an efficient CKM-assisted approximation to the original
problem~$\mathcal{P}_1$.

The computational complexity of the proposed algorithm is analyzed with respect to the number of users $K$ and the number of search points $M$ employed in the exhaustive search over $t$. Assuming that the generalized exponential integral $E_v(\cdot)$ can be evaluated with constant complexity, the evaluations of $\widehat{R}_c$ and $\{\widehat{R}_k\}_{k=1}^{K}$ require $\mathcal{O}(K^2)$ operations for a given $t$. Since the water-filling procedure only incurs $\mathcal{O}(K\log_2 K)$ complexity, the complexity per search point is dominated by the rate evaluations and scales as $\mathcal{O}(K^2)$. Therefore, the overall computational complexity of the proposed algorithm is $\mathcal{O}(MK^2)$.

\vspace{-2mm}
\section{Simulation Results}
\label{result}

In this section, simulations are conducted to evaluate the accuracy of the proposed bounds and the improvement in ergodic performance. Unless otherwise specified, the system parameters are configured as follows. The BS is equipped with \( N_t = 8 \) transmit antennas and simultaneously serves \( K = 4 \) single-antenna vehicles. The total transmit power is \( P = 20~\text{dBm} \). The vehicle speed is \( v_k = 180~\text{km/h} \). The carrier frequency is \( f_c = 5.9~\text{GHz} \)~\cite{FCC2020}, with a bandwidth of \( B = 10~\text{MHz} \). The noise power density is \( \sigma_0^2 = -174~\text{dBm/Hz} \), resulting in a total noise power of \( \sigma^2 = \sigma_0^2 B \). The CSI acquisition delay is \( T = 0.2~\text{ms} \), and blocklengths are set to \( l_c = l_k = 300 \). The block error rate for all data streams is \( \beta_{c,k} = \beta_{p,k} = 10^{-6} \)~\cite{9390169}.

\begin{figure}[t]
    \centering
    \begin{subfigure}[t]{0.45\textwidth}
        \centering
        \includegraphics[width=\linewidth]{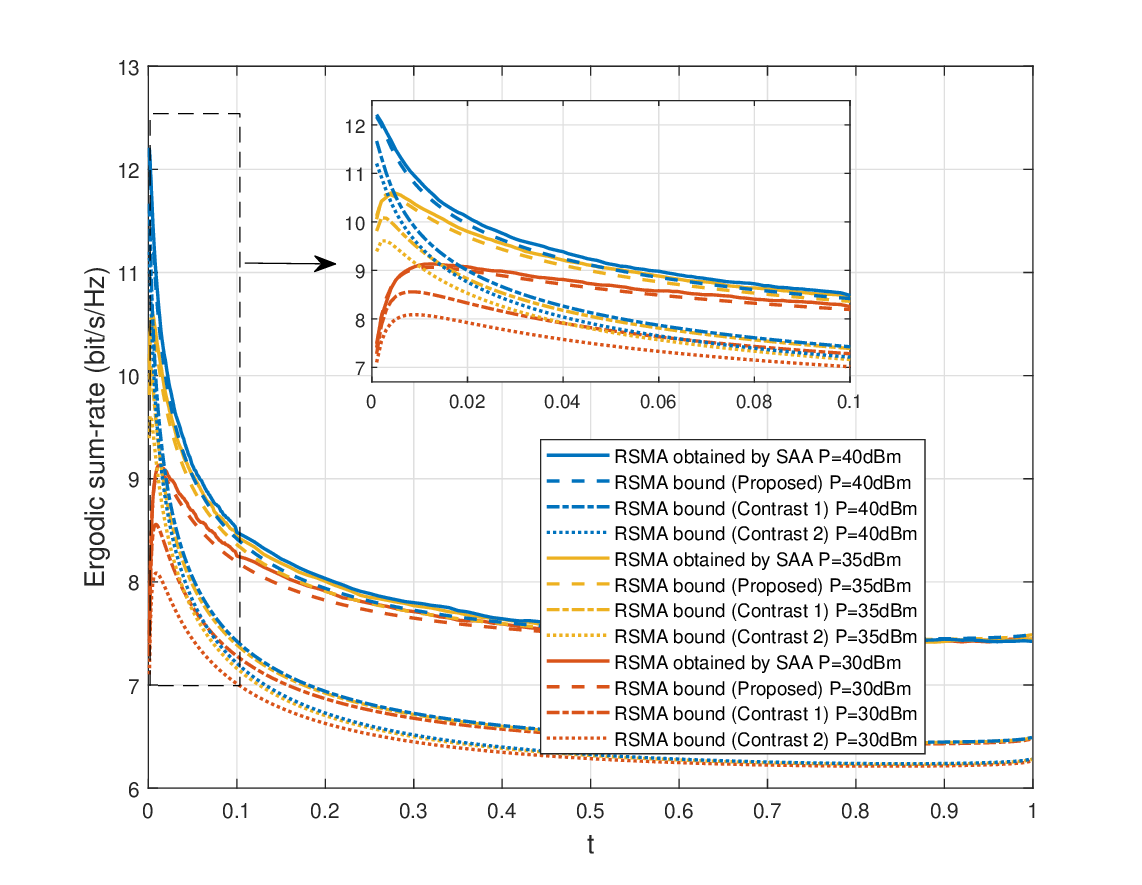}

        \caption{$N_t=8$, $K=4$, $\varepsilon=0.8$.}
        \label{fig_bound_a}
    \end{subfigure}

    \begin{subfigure}[t]{0.45\textwidth}
        \vspace{1mm}
        \centering
        \includegraphics[width=\linewidth]{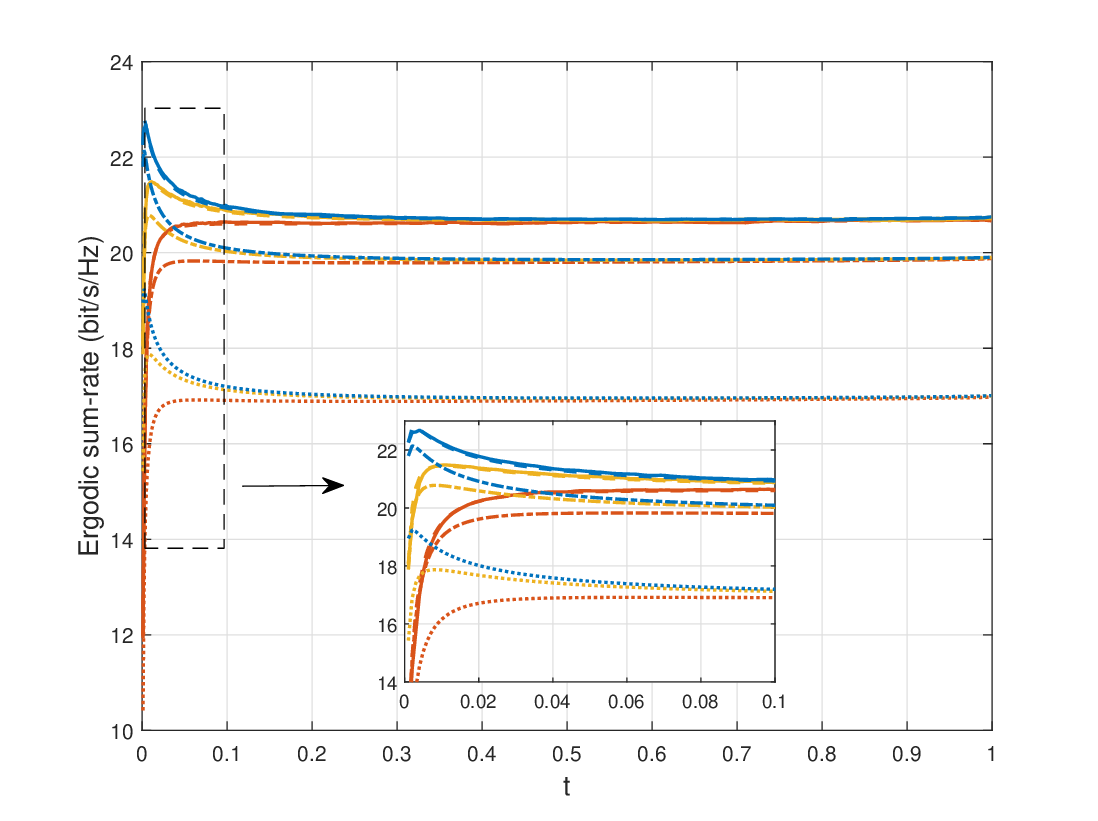}

        \caption{$N_t=32$, $K=8$, $\varepsilon=0.8$.}
        \label{fig_bound_b}
    \end{subfigure}
    \vspace{3mm}
    \caption{SAA, proposed method, and benchmark bounds.}
    \label{fig_bound2}
\end{figure}

We first evaluate the accuracy of the derived lower bound by comparing it with the simulated ergodic sum rate obtained via Monte Carlo simulation with $10^4$ channel realizations. Fig.~\ref{fig_bound2} presents the ergodic sum-rate results obtained from Monte Carlo simulation, the proposed closed-form lower bound, and two benchmark closed-form bounds, denoted as Contrast~1 and Contrast~2. Specifically, Contrast~1 and Contrast~2 correspond to two different closed-form ergodic rate lower bounds with the private stream derived in~\cite{11159579} and~\cite{10478577}, respectively. For these comparisons, we set $\zeta_{k}=1000$ and $\mu_k=1/K$, $\forall k \in \mathcal{K}$. Notice that our optimization targets the minimum user rate. Examining the sum rate provides a stricter test, as it reflects the total contribution from both the common and private streams. 

From Fig.~\ref{fig_bound2}, it can be observed that the proposed bound provides the tightest approximation to the actual ergodic sum rate among all considered methods. This is because, unlike the benchmark bounds that directly lower-bound the infinite-blocklength Shannon-capacity term $C(\cdot)$, the proposed approach approximates $C(\cdot)$ using a Gamma distribution, resulting in a tighter overall bound under the finite-blocklength regime and ensuring the reliability of the optimization in \eqref{eqn:problem_p0}. By comparing Fig.~\ref{fig_bound_a} and Fig.~\ref{fig_bound_b}, it is further observed that the proposed bound remains consistently tight and approaches the Monte Carlo results more closely as the number of transmit antennas increases, reaping the benefit of increased spatial degrees of freedom. 
In contrast, both benchmark bounds remain consistently looser than the proposed bound across all considered configurations.

\begin{figure}[t]
    \centering
    \includegraphics[width=0.45\textwidth]{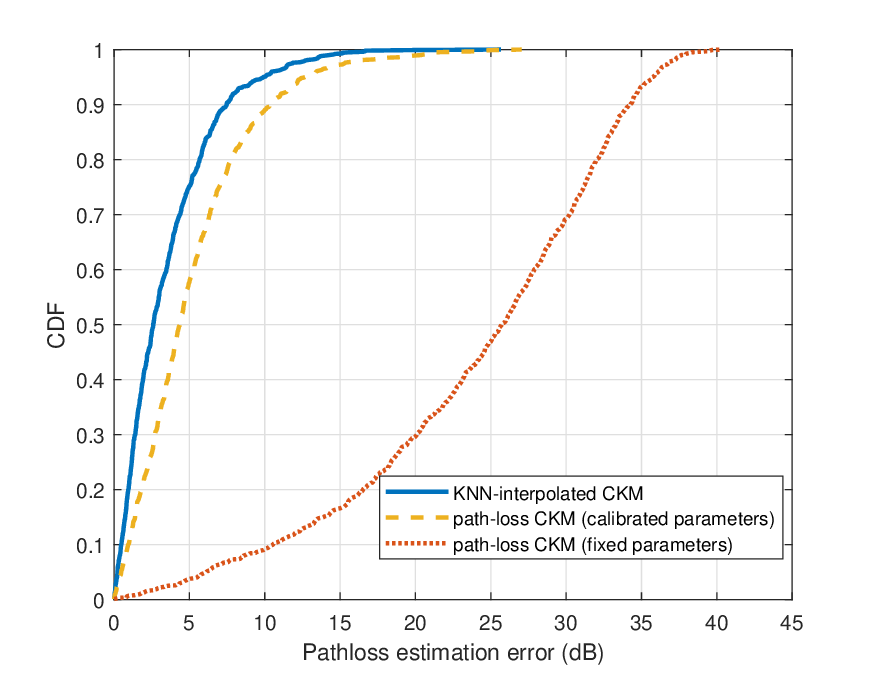}
    \caption{CDF of absolute estimation error for different CKM generation methods}
    \label{fig_CDF}
\end{figure}

In Fig.~\ref{fig_CDF}, the cumulative distribution functions (CDFs) of the absolute estimation errors in dB for different CKM generation methods are compared with the ground-truth CKM in Fig.~\ref{CKM_total}(a). 
The KNN-interpolated CKM is shown in Fig.~\ref{CKM_total}(b).
For the analytical path-loss CKM, two parameter settings are considered. The first uses calibrated parameters obtained from measured path losses at reference distances of $100~\mathrm{m}$ and $900~\mathrm{m}$ based on \eqref{pathloss_model}. The second adopts fixed parameters with $p_0=30$ and $\alpha=3.7$, corresponding to a relatively severe urban propagation condition with strong large-scale attenuation. The KNN-interpolated CKM exhibits the leftmost CDF curve, indicating the lowest estimation errors among the considered methods. The calibrated path-loss CKM achieves moderate performance, while the fixed-parameter path-loss CKM exhibits the largest estimation errors. 
These results show that the KNN-interpolated CKM provides more accurate path-loss estimates than the two analytical path-loss models in the considered scenario.

From the channel knowledge maps shown in Fig.~\ref{CKM_total}, four representative sampling points at coordinates $(50,0)$, $(250,0)$, $(450,0)$, and $(650,0)$ are selected, corresponding to four high-speed vehicles approximately 200~m apart along the road segment and consistently applied to all three CKMs for a fair comparison. Specifically, the resulting large-scale fading values (in dB) at these points are $\{76.767, 90.142, 106.507, 101.114\}$ for the ground-truth CKM obtained from dense sampling, $\{79.213, 90.688, 102.581, 99.423\}$ for the KNN-interpolated CKM, and $\{78.938, 96.372, 101.157, 104.017\}$ for the path-loss CKM. In the subsequent min-rate performance evaluation, resource allocation is performed based on the three CKMs, whereas the actual achievable rates are computed using the ground-truth CKM to ensure that the reported performance reflects the true channel conditions.

We consider three multiple access schemes—RSMA, SDMA, and NOMA. Their performance is evaluated under each of the three CKMs, resulting in a total of nine configurations. Specifically, the schemes are implemented as follows:

\begin{itemize}
    \item RSMA: Rate-splitting scheme proposed in this work, with optimized power allocation and common rate splitting based on the derived closed-form expressions. The global power coefficient $t$ is exhaustively searched over $t = 0.001 + (1-0.001)\, x^3$, with $x$ consisting of 400 uniformly spaced points in $[0,1)$.

    \item SDMA: Space-division multiple access with private power coefficients $\mu_k$ exhaustively searched via Monte Carlo in $[0,1]$ with step size $0.025$.

    \item NOMA: Non-orthogonal multiple access using inter-group spatial division multiplexing. Intra- and inter-group power allocations are exhaustively searched via Monte Carlo with a step size of 0.025, using one layer of SIC per user and proximity-based pairing.
\end{itemize}

\begin{figure}[t]
    \centering
    \includegraphics[width=0.45\textwidth]{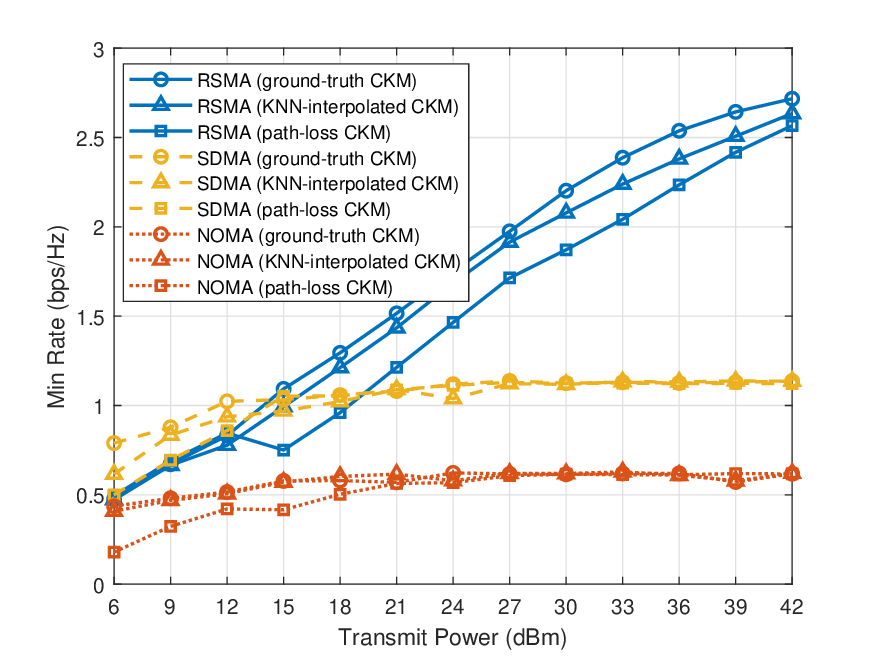}
    \caption{Ergodic min-rate versus transmit power}
    \label{fig_MinRate_vs_SNR}
\end{figure}

Fig.~\ref{fig_MinRate_vs_SNR} presents the min-rate performance of the considered transmission schemes as a function of the transmit power. In the high transmit power region, RSMA outperforms both SDMA and NOMA in terms of the minimum rate due to its flexible interference management and partial decoding capability. Among RSMA schemes, the ground-truth CKM yields the best performance, followed by the KNN-interpolated CKM, while the path-loss-based CKM performs the worst, exhibiting a clear performance stratification. 
These performance gaps highlight the importance of CKM accuracy for the proposed allocation scheme. Inaccurate large-scale channel information can lead to mismatched common/private power splitting and common-rate allocation, thereby reducing the minimum user rate. The resulting performance differences remain observable over the evaluated high-power range.
At low transmit power, RSMA effectively reduces to SDMA with equal private power allocation and consequently underperforms SDMA with optimized power allocation. For SDMA, the ground-truth CKM achieves the highest minimum rate, followed by the KNN-interpolated CKM, whereas the path-loss-based CKM yields the lowest performance. For NOMA, both the ground-truth and KNN-interpolated CKMs preserve the correct strong–weak user ordering and therefore incur no noticeable degradation; In contrast, the path-loss-based CKM occasionally misidentifies the strong–weak ordering of users, leading to a discernible rate loss. As the transmit power increases, the influence of CKM inaccuracies on SDMA and NOMA becomes negligible because their optimal power allocations progressively converge to more balanced patterns across users or user groups, thereby reducing sensitivity to large-scale fading errors.

\begin{figure}[t]
    \centering
    \includegraphics[width=0.45\textwidth]{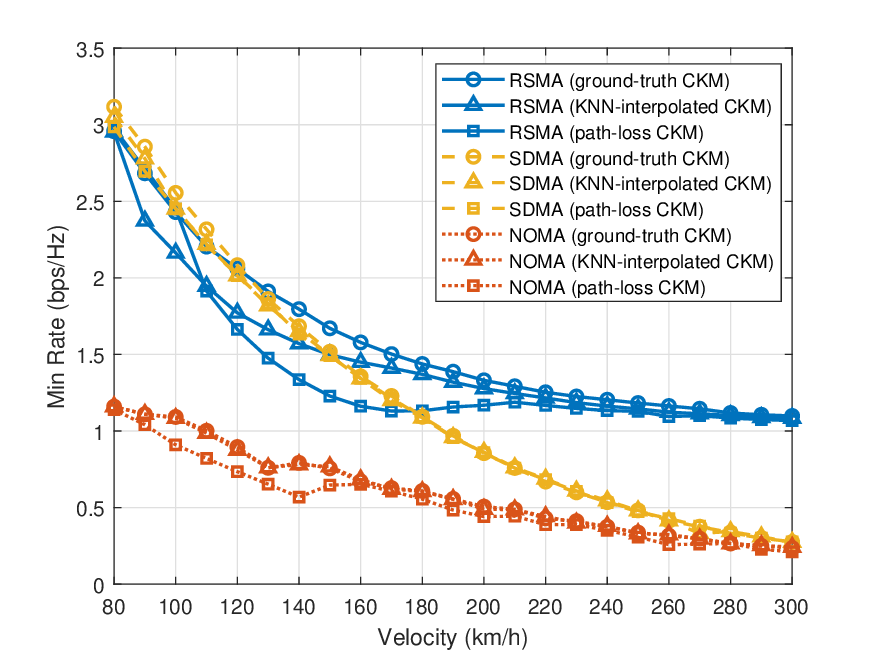}
    \caption{Ergodic min-rate versus vehicle velocity}
    \label{fig_MinRate_vs_Velocity}
\end{figure}

Fig.~\ref{fig_MinRate_vs_Velocity} illustrates the ergodic min-rate performance of the considered schemes as a function of vehicle velocity. As the velocity increases, RSMA exhibits stronger robustness, whereas the performance of SDMA and NOMA degrades due to their limited ability to fully suppress inter-user or inter-group interference under rapidly varying channels. In addition, the performance gap caused by different CKM maps narrows for all schemes at higher velocities. For RSMA, this is because the system increasingly allocates power to the common stream to mitigate interference, which reduces the contribution of private-stream inaccuracies. As a result, the common-stream splitting across different CKM maps converges toward a nearly uniform allocation. Similarly, SDMA and NOMA tend toward more balanced power allocation across users or user groups as velocity increases, which further diminishes the impact of CKM map discrepancies.


\begin{figure}[t]
    \centering
    \begin{subfigure}[t]{0.45\textwidth}
        \centering
        \includegraphics[width=\linewidth]{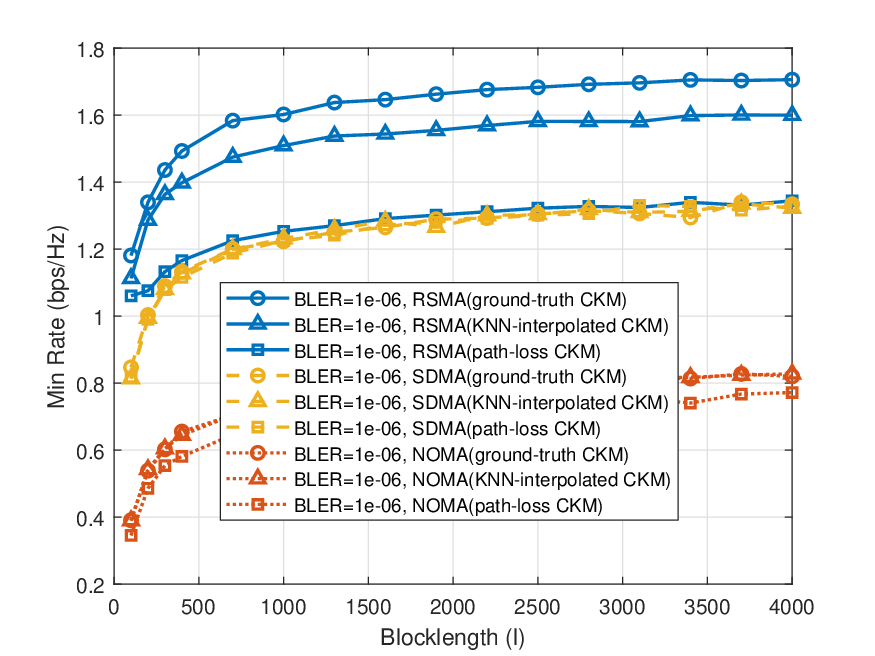}
        \caption{BLER = $10^{-6}$.}
        \label{fig_MinRate_1e-6}
    \end{subfigure}
    \begin{subfigure}[t]{0.45\textwidth}
        \centering
        \vspace{1mm}
        \includegraphics[width=\linewidth]{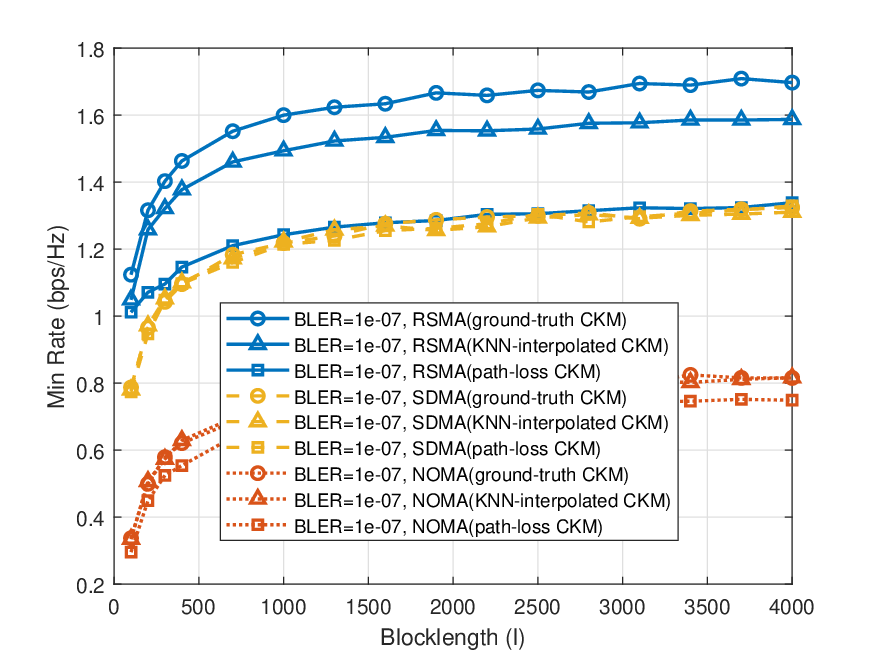}
        \caption{BLER = $10^{-7}$.}
        \label{fig_MinRate_1e-7}
    \end{subfigure}
    \vspace{2mm}
    \caption{Ergodic min-rate versus blocklength under different BLERs}
    \label{fig_MinRate_vs_Blocklength}
\end{figure}

Fig.~\ref{fig_MinRate_vs_Blocklength} depicts the minimum achievable rate versus the blocklength under two BLER targets.
For a given BLER, the min-rate increases monotonically with the blocklength.
By comparing Fig.~\ref{fig_MinRate_1e-6} and Fig.~\ref{fig_MinRate_1e-7}, it is observed that tightening the BLER constraint from $10^{-6}$ to $10^{-7}$ leads to a moderate reduction in the achievable min-rate for all considered schemes.
Across both BLER settings, RSMA consistently achieves the highest minimum rate, indicating its ability to meet the target minimum rate with shorter blocklengths and stricter reliability requirements.
Moreover, RSMA exhibits noticeable sensitivity to the CKM accuracy due to its reliance on precise large-scale channel information for common-stream rate splitting, whereas the performance of SDMA is nearly invariant to the CKM choice, and NOMA shows a discernible gain only when compared with the path-loss-based CKM.
Overall, SDMA and NOMA are less affected by CKM variations, as system performance in the considered high-mobility regime is primarily constrained by Doppler-induced inter-user interference rather than large-scale channel knowledge inaccuracies.

\section{Conclusion}
\label{conclusion}
This paper has investigated CKM-assisted FBL RSMA for high-mobility vehicular communications under outdated CSIT. The proposed framework has exploited location-dependent large-scale channel knowledge provided by CKM to guide the common/private power split and common-rate allocation. A tight closed-form lower bound on the private-stream ergodic rate has been derived, enabling a tractable max-min rate optimization under FBL constraints. Numerical evaluations have further confirmed that the proposed bound closely approximates the Monte Carlo results and is consistently tighter than the benchmark bounds across the considered configurations.
The results have also shown that, among the considered schemes, RSMA has provided stronger robustness to channel aging than SDMA and NOMA in high-mobility scenarios. At the same time, its performance has remained sensitive to inaccurate large-scale channel knowledge, since CKM errors affect both the common-stream bottleneck rate and the resulting common-rate allocation. These results have highlighted the complementary roles of RSMA and CKM: RSMA has provided greater robustness against residual multi-user interference associated with outdated CSIT, while accurate CKM information has improved channel-aware resource allocation of RSMA. As a future direction, CKM reliability can be explicitly incorporated into RSMA resource allocation, enabling uncertainty-aware power splitting and common-rate allocation for more robust high-mobility short-packet communications.

\appendices

\section{Proof Of Lower Bound for $R_c$}

The proof consists of three main steps.

First, by introducing the auxiliary random variable $Y_{c,k}$, the minimum common-stream SINR is approximated as
\[
\min_{k\in\mathcal K}\Gamma_{c,k}
\approx
P(1-t)\min_{k\in\mathcal K}\left(\zeta_kY_{c,k}\right),
\]
where
$Y_{c,k}
=
\frac{|{\mathbf h}_k^H{\mathbf p}_c|^2}
{1+\frac{Pt\zeta_k}{K}
\sum_{j\in\mathcal K}|{\mathbf h}_k^H{\mathbf p}_j|^2}.$ The accuracy of this approximation has been verified through numerical experiments, as demonstrated in Fig. 2 of~\cite{11159579}.

The interference term $\sum_{j\in\mathcal K}|{\mathbf h}_k^H{\mathbf p}_j|^2$ is approximated as $\widetilde X\sim \mathrm{Gamma}(\widetilde D,\widetilde\theta)$ by applying the Gamma approximation in~\eqref{gammajinsi}, where $\widetilde D=\frac{\left[\varepsilon^2(N_t+1)+(1-2\varepsilon^2)K\right]^2}{\varepsilon^4(N_t+1)+(1-2\varepsilon^2)K}$ and $\widetilde\theta=\frac{\varepsilon^4(N_t+1)+(1-2\varepsilon^2)K}{\varepsilon^2(N_t+1)+(1-2\varepsilon^2)K}$.

Accordingly, $Q_{k2}=\frac{\widetilde X}{\widetilde\theta}\sim\Gamma(\widetilde D,1)$ and $Q_{k1}=|{\mathbf h}_k^H{\mathbf p}_c|^2\sim\Gamma(1,1)$. Based on this approximation and the independence assumption between $Q_{k1}$ and $Q_{k2}$, $Y_{c,k}$ can be approximated by $\widetilde Y_{c,k}$, whose CDF is given by

\[
F_{\widetilde Y_{c,k}}(y)
=
1-
\frac{e^{-y}}
{\left(1+\frac{P\widetilde\theta t\zeta_k}{K}y\right)^{\widetilde D}}.
\]
By further assuming independent users, $\widetilde\Gamma_c=P(1-t)\min_{k\in\mathcal K}(\zeta_k\widetilde Y_{c,k})$ serves as an approximation of $\min_{k\in\mathcal K}\Gamma_{c,k}$, and its CDF is given by

\[
F_{\widetilde\Gamma_c}(y)
=
1-
\frac{e^{-C_1y}}
{(1+C_2y)^D},
\]
where
$C_1=\sum_{k=1}^{K}\frac{1}{\zeta_kP(1-t)}$,
$C_2=\frac{\widetilde\theta t}{K(1-t)}$,
and
$D=\lfloor\widetilde DK\rceil$.

Next, using the identity
\[
\mathbb E[\log_2(1+X)]
=
\frac{1}{\ln2}
\int_0^\infty
\frac{1-F_X(x)}{1+x}{\rm{d}}x,
\]
together with the above CDF yields
\[
\mathbb E[C(\widetilde\Gamma_c)]
=
\frac{e^{C_1/C_2}}{\ln2}\Psi(D).
\]
Similarly,
\[
\mathbb E[\widetilde\Gamma_c]
=
\int_0^\infty
\left(1-F_{\widetilde\Gamma_c}(x)\right){\rm{d}}x
=
\frac{e^{C_1/C_2}}{C_2}
E_D\!\left(\frac{C_1}{C_2}\right).
\]

Finally, since $\sqrt{V(x)}$ is concave,
Jensen's inequality gives
\[
\mathbb E\!\left[\sqrt{V(\widetilde\Gamma_c)}\right]
\le
\sqrt{V\!\left(\mathbb E[\widetilde\Gamma_c]\right)},
\]
which leads to the lower bound in \eqref{Rc_lower_bound_final} after substituting the above closed-form expressions.

\section{Proof Of Lemma \ref{log1plusZ}}

Let $\widetilde Z \sim \mathrm{Gamma}(\widetilde D_Z,\widetilde \theta_Z)$ with integer
$\widetilde D_Z \ge 1$ and $\widetilde \theta_Z > 0$. We prove Lemma~\ref{log1plusZ} using a recursive approach: starting from the base case of $\widetilde D_Z = 1$, we derive a closed-form expression and then establish a recursive relation for the general integer $\widetilde D_Z \ge 1$, which is finally summed to obtain the desired result.

\textbf{Step 1: Base case ($\widetilde Z_1$).}  
Let $\widetilde Z_1 \sim \mathrm{Gamma}(1, \widetilde \theta_Z)$. 
When the shape parameter equals one, $\widetilde Z_1$ reduces to an exponential random variable with probability density function
$f_{\widetilde Z_1}(z)=\frac{1}{\widetilde \theta_Z}e^{-z/\widetilde \theta_Z}$.
Then,
\begin{align}
\mathbb{E}[\log_2(1+\widetilde Z_1)]
&= \frac{1}{\widetilde \theta_Z \ln 2}
\int_0^\infty e^{-z/\widetilde \theta_Z}\ln(1+z) {\rm{d}}z \nonumber\\
&= \frac{1}{\widetilde \theta_Z \ln 2}
\int_0^\infty \ln(1+z){\rm{d}}\left(-\widetilde \theta_Z e^{-z/\widetilde \theta_Z}\right) \nonumber\\
&= \frac{1}{\ln 2}
\int_0^\infty \frac{e^{-z/\widetilde \theta_Z}}{1+z}{\rm{d}}z,
\label{eq:basecase-ibp}
\end{align}
where \eqref{eq:basecase-ibp} follows from integration by parts with vanishing boundary terms.
By the change of variable $u = 1+z$, the above integral can be rewritten as
\begin{equation}
\mathbb{E}[\log_2(1+\widetilde Z_1)]
= \frac{e^{1/\widetilde \theta_Z}}{\ln 2}
\int_1^\infty \frac{e^{-u/\widetilde \theta_Z}}{u}{\rm{d}}u
= \frac{e^{1/\widetilde \theta_Z}}{\ln 2}
E_1\!\left(\frac{1}{\widetilde \theta_Z}\right),
\label{eq:ergodic-rate-base}
\end{equation}
where \(
E_v(x) = \int_1^\infty t^{-v} e^{-t x} {\rm{d}}t
\).

\textbf{Step 2: Recursive relation.}

For any nonnegative random variable $\widetilde Z$ with CDF $F_{\widetilde Z}(z)$, integration by parts gives
\begin{equation}
\mathbb{E}[\log_2(1+\widetilde Z)]
= \frac{1}{\ln 2} \int_0^\infty \frac{1 - F_{\widetilde Z}(z)}{1+z}{\rm{d}}z.
\label{eq:log_expect_cdf}
\end{equation}

Now consider $\widetilde Z_d \sim \mathrm{Gamma}(d, \widetilde \theta_Z)$ with integer shape parameter $d \ge 1$. The survival function is
\begin{equation}
1 - F_{\widetilde Z_d}(z)
= \frac{\Gamma(d, z/\widetilde \theta_Z)}{(d-1)!},
\label{eq:gamma_survival}
\end{equation}
where $\Gamma(s,x)=\int_x^\infty t^{s-1}e^{-t}{\rm{d}}t$.

Substituting this into the expectation formula, we have
{\small
\begin{multline}
\mathbb{E}[\log_2(1+\widetilde Z_{d+1})]
- \mathbb{E}[\log_2(1+\widetilde Z_d)] \\
= \frac{1}{\ln 2} \int_0^\infty \frac{\Gamma(d+1, z/\widetilde \theta_Z)}{d!\,(1+z)}{\rm{d}}z
- \frac{1}{\ln 2} \int_0^\infty \frac{\Gamma(d, z/\widetilde \theta_Z)}{(d-1)!\,(1+z)} {\rm{d}}z.
\label{eq:gamma-substitution}
\end{multline}
}

Using the recurrence relation of the upper incomplete Gamma function,
\[
\Gamma(d+1, x) = d \, \Gamma(d,x) + x^d e^{-x},
\]
and applying it to the first integral in \eqref{eq:gamma-substitution}, the difference simplifies to
\begin{multline}
\mathbb{E}[\log_2(1+\widetilde Z_{d+1})]
- \mathbb{E}[\log_2(1+\widetilde Z_d)] \\
= \frac{1}{\ln 2} \int_0^\infty \frac{(z/\widetilde \theta_Z)^d e^{-z/\widetilde \theta_Z}}{d!\,(1+z)} {\rm{d}}z.
\label{eq:diff-integral}
\end{multline}

With the variable substitution $z = \widetilde \theta_Z u$, ${\rm{d}}z = \widetilde \theta_Z {\rm{d}}u$, \eqref{eq:diff-integral} becomes

\begin{equation}
\label{integral-one}
\frac{1}{\ln 2} \int_0^\infty \frac{(z/\widetilde \theta_Z)^d e^{-z/\widetilde \theta_Z}}{d!\,(1+z)} {\rm{d}}z
= \frac{\widetilde \theta_Z}{\ln 2 \, d!} \int_0^\infty \frac{u^d e^{-u}}{1 + \widetilde \theta_Z u} {\rm{d}} u
\end{equation}

Using the Laplace transform, the denominator in the above integral can be expanded and expressed as
\[
\frac{1}{1 + \widetilde \theta_Z u} = \int_0^\infty e^{-(1+\widetilde \theta_Z u)t} {\rm{d}}t,
\]
Accordingly, \eqref{integral-one} can be expressed as a double integral, and by interchanging the order of integration, we obtain
\begin{align}
\frac{\widetilde \theta_Z}{\ln 2 \, d!} \int_0^\infty & 
\frac{u^d e^{-u}}{1 + \widetilde \theta_Z u} {\rm{d}}u \notag\\
&= \frac{\widetilde \theta_Z}{\ln 2 \, d!} \int_0^\infty e^{-t} 
\Biggl( \int_0^\infty u^d e^{-u(1+\widetilde \theta_Z t)} {\rm{d}}u \Biggr) {\rm{d}}t
\label{eq_double_integral}
\end{align}

Evaluating the inner integral using the Gamma function
\(
\Gamma(d+1) = \int_0^\infty t^d e^{-t} {\rm{d}}t, \quad d \in \mathbb{Z}_{\ge 0},
\)
which reduces to the factorial $d!$. We make the change of variable
$v = u (1+\widetilde \theta_Z t)$, yielding

\begin{align}
\int_0^\infty u^d e^{-u(1+\widetilde \theta_Z t)} {\rm{d}}u
&= \frac{1}{(1+\widetilde \theta_Z t)^{d+1}} \int_0^\infty v^d e^{-v} {\rm{d}}v \notag\\
&= \frac{\Gamma(d+1)}{(1+\widetilde \theta_Z t)^{d+1}} \notag\\
&= \frac{d!}{(1+\widetilde \theta_Z t)^{d+1}},
\label{eq_double_djiecheng}
\end{align}

Thus, substituting \eqref{eq_double_djiecheng} into \eqref{eq_double_integral} becomes
\begin{equation}
\frac{\widetilde \theta_Z}{\ln 2} \int_0^\infty \frac{e^{-t}}{(1+\widetilde \theta_Z t)^{d+1}} {\rm{d}}t.
\end{equation}

Then, applying the change of variable $1+\widetilde \theta_Z t = o$, we have 
\begin{align}
\frac{\widetilde \theta_Z}{\ln 2} \int_1^\infty & 
\frac{e^{-(o-1)/\widetilde \theta_Z}}{o^{d+1}} \frac{{\rm{d}}o}{\widetilde \theta_Z} \notag\\
&= \frac{1}{\ln 2} \int_1^\infty \frac{e^{1/\widetilde \theta_Z} e^{-o/\widetilde \theta_Z}}{o^{d+1}} {\rm{d}}o \notag\\
&= \frac{e^{1/\widetilde \theta_Z}}{\ln 2} \int_1^\infty \frac{e^{-o/\widetilde \theta_Z}}{o^{d+1}} {\rm{d}}o
\end{align}

Finally, recognizing the integral as the generalized exponential integral
\(
E_v(x) = \int_1^\infty t^{-v} e^{-t x} {\rm{d}}t
\), we obtain the recursive relation
\begin{equation}
\mathbb{E}[\log_2(1+\widetilde Z_{d+1})]
- \mathbb{E}[\log_2(1+\widetilde Z_d)]
= \frac{e^{1/\widetilde \theta_Z}}{\ln 2} E_{d+1}\!\left(\frac{1}{\widetilde \theta_Z}\right),
\label{eq:gamma-recursive}
\end{equation}which establishes the recursive step used for the telescoping sum.

\textbf{Step 3: Telescoping sum.}  
For $\widetilde Z \sim \mathrm{Gamma}(\widetilde D_Z,\widetilde \theta_Z)$ with integer
$\widetilde D_Z \ge 2$ and $\widetilde \theta_Z > 0$, summing the recursive relation
\eqref{eq:gamma-recursive} from $d=1$ to $d=\widetilde D_Z-1$ and using the base case
in \eqref{eq:ergodic-rate-base} yields
\begin{equation}
\mathbb{E}[\log_2(1+\widetilde Z)]
= \frac{e^{1/\widetilde \theta_Z}}{\ln 2}
\sum_{i=0}^{\widetilde D_Z-1}
E_{i+1}\!\left(\frac{1}{\widetilde \theta_Z}\right).
\label{eq:ergodic-rate-general}
\end{equation}
This expression also holds for $\widetilde D_Z = 1$, thereby completing the proof.

\ifCLASSOPTIONcaptionsoff
  \newpage
\fi

\bibliographystyle{IEEEtran}
\balance
\bibliography{reference.bib}

\end{document}